\documentclass{article}

\usepackage{amsmath}
\usepackage{amssymb}
\usepackage{amsthm}
\usepackage{hyperref}

\pdfmapline{+upml upml <upml.pfb}

\DeclareFontFamily{U}{upml}{}
\DeclareFontShape{U}{upml}{m}{n}{<-6>upml<6-8>upml<8->upml}{}
\DeclareSymbolFont{UPML}{U}{upml}{m}{n}

\DeclareMathSymbol{\blwedge}{\mathrel}{UPML}{"63}
\DeclareMathSymbol{\blvee}{\mathrel}{UPML}{"64}
\DeclareMathSymbol{\blRightarrow}{\mathrel}{UPML}{"69}
\DeclareMathSymbol{\blLeftrightarrow}{\mathrel}{UPML}{"65}
\DeclareMathSymbol{\blneg}{\mathord}{UPML}{"6E}
\DeclareMathSymbol{\blotimes}{\mathbin}{UPML}{"6F}

\def\logand{\blotimes}

\def\loginf{\blwedge}
\def\logsup{\blvee}

\newcounter{global}
\theoremstyle{definition}

\newtheorem{assumption}[global]{Assumption}
\theoremstyle{plain}
\newtheorem{theorem}[global]{Theorem}
\newtheorem{lemma}[global]{Lemma}
\newtheorem{corollary}[global]{Corollary}
\newtheoremstyle{note}{}{}{}{}{\itshape}{.}{.5em}{}
\theoremstyle{note}
\newtheorem{remark}{Remark}%
\renewcommand\section{%
  \@startsection {section}{1}{\z@}%
  {-3.5ex \@plus -1ex \@minus -.2ex}%
  {2.3ex \@plus.2ex}%
  {\normalfont\large\bfseries}}
\makeatletter
\let\logand=\loginf
\let\Rightarrow=\logimp
\let\Leftrightarrow=\logequ
\let\logtnorm=\blotimes
\def\itm#1{{\rm(\textit{\romannumeral#1})}}
\def\Luk{\ensuremath{\text{\textbf{\upshape{\L}}}}}
\def\Gog{\ensuremath{\mathbf{\Pi}}}
\def\RatL{\ensuremath{[0,1]_\mathbb{Q}}}

\def\C{\@ifnextchar[{\@withC}{\@withoutC}}
\def\@withC[#1]{\ensuremath{\boldsymbol{c}_{#1}\@ifnextchar\bgroup{\@witharg}{\relax}}}
\def\@withoutC{\ensuremath{\boldsymbol{c}\@ifnextchar\bgroup{\@witharg}{\relax}}}
\def\@witharg#1{\@leftpar #1\@rightpar}
\makeatother

\begin{document}

\title{Rational fuzzy attribute logic}

\date{\normalsize%
  Dept. Computer Science, Palacky University, Olomouc}

\author{Vilem Vychodil\footnote{%
    e-mail: \texttt{vychodil@binghamton.edu},
    phone: +420 585 634 705,
    fax: +420 585 411 643}}

\maketitle

\begin{abstract}
  We present a logic for reasoning with if-then formulas which involve
  constants for rational truth degrees from the unit interval. We introduce
  graded semantic and syntactic entailment of formulas. We prove
  the logic is complete in Pavelka style and depending on the choice of
  structure of truth degrees, the logic is a decidable fragment of
  the Rational Pavelka logic (RPL) or the Rational Product Logic (R$\Pi$L).
  We also present a characterization of the entailment based
  on least models and study related closure structures.
\end{abstract}

\section{Introduction}\label{sec:intro}
In this paper, we are interested in a logic for reasoning with if-then rules
describing dependencies between graded attributes. A graded attribute may be seen
as a propositional variable which may be assigned a truth degree coming from a scale
of degrees which includes intermediate degrees of truth. In this paper, we use
particular complete residuated lattices on the real unit interval as the scales of
truth degrees. As a consequence, if $p$ is a propositional variable and $e(p)$
denotes its truth value under the evaluation $e$, we admit $e(p) \in [0,1]$ with
the possibility of $0 < e(p) < 1$. The meaning of the degrees assigned to
propositional variables conforms to the meaning in mathematical fuzzy
logics~\cite{CiHaNo1,CiHaNo2,Got:Mfl,Haj:MFL}, i.e.,
the meaning is \emph{comparative}: $e(p) < e(q)$ means that propositional variable
$p$ is (strictly) less true than $q$ under the evaluation $e$.
Let us stress at this point that the degrees we use are not and shall not be
interpreted as degrees of belief or evidence, cf. ``the frequentist's temptation''
in \cite{Haj:MFL} and also \cite{HaPa:Adfl}.

Admitting the intermediate degrees of truth may be valuable in situations
where attributes represented by propositional variables may not be assigned either
of the classic truth degrees $0$ (falsity) and $1$ (truth). This may occur in
situations where human perception and judgment is involved. For instance,
if we ask a person whether \$400,000 is a high price for a house in
a particular neighborhood, the person may hesitate to describe the price
as being (strictly) \emph{high} or (strictly) \emph{not high}.
Indeed, the person may feel that a sentence
``The price of the house sold at \$400,000 is high''
is true in a general way, but is not completely true. This fact may be
captured by assuming that the attribute ``price is high'' represented by
a propositional variable $p$ is assigned a truth degree $e(p)$ which is close
to $1$ but strictly less than $1$ and we may further be interested in reasoning
with this fact. Formal treatment of reasoning with such facts is the main
subject of mathematical fuzzy logics and is also the main subject of this paper.

In this paper we focus on reasoning with formulas formalizing if-then dependencies
between graded attributes. Namely, given an evaluation of propositional variables
$p_1,p_2,\ldots,q_1,q_2,\ldots$, we formalize (the semantics of) rules of the form
\begin{quote}
  if
  ($p_1$ is true at least to degree $a_1$ and ${\cdots}$ and \\
  $\phantom{\text{if (}}$%
  $p_n$ is true at least to degree $a_n$), \\
  then 
  ($q_1$ is true at least to degree $b_1$ and ${\cdots}$ and \\
  $\phantom{\text{then (}}$%
  $q_n$ is true at least to degree $b_n$).
\end{quote}
Moreover, we present a logical system for reasoning with such rules
and prove its soundness, completeness, and present its relationship to
existing logics. We call our logic \emph{rational} because we allow only
rational degrees from the unit interval to appear in the if-then rules.
This is a similar assumption as in the Rational Pavelka logic (RPL)
proposed by H\'ajek~\cite{Haj:Flah,Haj:MFL} which extends \L ukasiewicz logic
by constants for rational truth degrees and corresponding bookkeeping axioms.
In fact, when the standard \L ukasiewicz algebra is used as the structure
of truth degrees, our logic can be seen as a particular fragment of RPL.
Analogously, when the standard Goguen (product) algebra is used as the structure
of truth degrees, our logic becomes a fragment of the Rational Product Logic
(R$\Pi$L) which has been proposed by Esteva, Godo, H\'ajek, and Navara
in~\cite{EsGoHaNa:Rflin}.

Another common attribute of our logic and RPL (or R$\Pi$L) is that we
consider degrees of entailment on both the semantic and syntactic levels.
That is, for a formula $\varphi$ and a theory $\Sigma$, we consider
a degree $||\varphi||_\Sigma \in [0,1]$ to which $\varphi$ is semantically
entailed by $\Sigma$ and in general, we may have $0 < ||\varphi||_\Sigma < 1$.
More importantly, we also consider a degree $|\varphi|_\Sigma$ to which $\varphi$
is provable by $\Sigma$ and, again, we may have $0 < |\varphi|_\Sigma < 1$.
Logics with graded notions of provability were first investigated by
Pavelka~\cite{Pav:Ofl1,Pav:Ofl2,Pav:Ofl3} who was inspired by ideas in the
influential paper by Goguen~\cite{Gog:Lic}. In his seminal papers,
Pavelka proposed very general approach to abstract logics with
graded notions of provability and proved completeness of a propositional
logic based on the standard \L ukasiewicz algebra~\cite{KMP:TN}
as the structure of truth degrees. The completeness theorem states that
$|\varphi|_\Sigma = ||\varphi||_\Sigma$, i.e., the degrees of semantic
entailment coincide with the degrees of provability. This particular type
of completeness of multiple-valued logics has later become known as
the Pavelka completeness, cf.~\cite[Section 9.2]{Haj:MFL}.
The contribution of H\'ajek's RPL is in simplification of Pavelka's ideas,
including (i) the use of (constants for) only rational truth degrees and thus
keeping the language of the logic countable and (ii) considering proofs in the
ordinary sense instead of considering them as sequences of weighted formulas
as in the original approach~\cite{Pav:Ofl1,Pav:Ofl2,Pav:Ofl3} by Pavelka.
Further details and development of Pavelka-complete logics can be
found in \cite{Ger:FL,NoPeMo:MPFL} and the references therein.

The logic we present is complete in the Pavelka style and depending on the
choice of a structure of truth degrees (the standard \L ukasiewicz or
Goguen algebra), it may be seen as a decidable fragment of RPL or R$\Pi$L
which uses only formulas in a particular form. Our logic may be seen as logic
which falls into the category of fuzzy logics with rational constants for
truth degrees~\cite{CiEsGo:Lltc,EsGiGoNo:Atclct,EsGoNo:Fotfltc,SaCiEsGoNo:Plwc}
and develops our previous results on logic of fuzzy attribute implications
which were limited to finite structures of truth
degrees~\cite{BeVy:ICFCA,BeVy:ADfDwG} or utilized infinitary deduction rules
in order to ensure completeness in Pavelka style~\cite{BeVy:Falcrl,KuVy:Flprrai}.
In contrast to these previous results, the present approach shows logics with
finitary deduction rules and particular structures of truth degrees defined
on the real unit interval. In a broader sense, our paper is interested in
if-then rules which generalize analogous rules that appear in 
database systems as functional dependencies~\cite{Mai:TRD}),
logic programming~\cite{Lloyd84},
or data analysis~\cite{AgImSw:ASR,GaWi:FCA,Zak:Mnrar} as
attribute implications or association rules.

Our paper is structured as follows. In Section~\ref{sec:prelim}, we survey
preliminaries from structures of truth degrees used in this paper (the section
may be skipped by readers familiar with residuated lattices). In
Section~\ref{sec:rfal}, we describe our logic and present its completeness.
The proof of completeness together with further notes are presented
in Section~\ref{sec:details}. In Section~\ref{sec:clops}, we present
characterization of closure structures related to systems of models of theories.

\section{Preliminaries}\label{sec:prelim}
In our paper, we consider particular linear complete integral commutative
residuated lattices~\cite{GaJiKoOn:RL,DiWa} as the structures of degrees.
The structures are defined as general algebras
$\mathbf{L}= \langle L,\wedge,\vee,\otimes,\rightarrow,0,1\rangle$ of
type $\langle 2,2,2,2,0,0\rangle$, where
\begin{itemize}\parskip=0pt
\item[\itm{1}]
  $\langle L,\wedge,\vee,0,1 \rangle$ is a linear complete lattice~\cite{Bir:LT},
\item[\itm{2}]
  $\langle L,\otimes,1 \rangle$ is a commutative monoid, and
\item[\itm{3}]
  for all $a,b,c \in L$, we have
  $a \otimes b \leq c$ if{}f $a \leq b \rightarrow c$
\end{itemize}
Property \itm{3} is called the adjointness of $\otimes$ (truth function of
``fuzzy conjunction'') and $\rightarrow$ (truth function of ``fuzzy implication'');
note that $\leq$ denotes the linear order induced by $\wedge$, i.e.,
$a \leq b$ if{}f $a \wedge b = a$ (equivalently, $a \vee b = b$).

In our paper, we use linear complete residuated lattices defined on the real unit
interval which are given by left-continuous triangular
norms~\cite{KMP:TN} which play central role in the Basic Logic
(BL, see~\cite{Haj:MFL}) and Monoidal T-norm Logic (MTL, see \cite{EsGo:MTL}).
That is, $\langle L,\wedge,\vee,0,1 \rangle$ represents the real unit interval
with its natural ordering ($\wedge$ and $\vee$ coincides with operations of minima
and maxima, respectively), $\otimes$ is associative, commutative, neutral with
respect to $1$, and is left-continuous (distributive with respect to general suprema),
i.e., it satisfies
\begin{align}
  \textstyle\bigvee\{a \otimes b_i;\, i \in I\} &=
  a \otimes \textstyle\bigvee\{b_i;\, i \in I\}
\end{align}
for all $a \in [0,1]$ and $\{b_i \in [0,1];\, i \in I\}$.
Moreover, the corresponding (uniquely given) $\rightarrow$ which is adjoint
to $\otimes$ is then given by
\begin{align}
  a \rightarrow b &= \textstyle\bigvee\{c \in [0,1];\, a \otimes c \leq b\}.
\end{align}
Two structures which are most relevant for our investigation are the
so-called standard \L ukasiewicz and Goguen (product) algebras denoted
$\Luk$ and $\Gog$, respectively, where the multiplications and residua
are given by
\begin{align}
  a \otimes_\Luk b &= \max\{0,a+b-1\}, \\
  a \rightarrow_\Luk b &= \min\{1,1-a+b\}, 
\end{align}
and
\begin{align}
  a \otimes_\Gog b &= ab, \\
  a \rightarrow_\Gog b &=
  \begin{cases}
    1, & \text{if } a \leq b, \\
    \frac{b}{a}, &\text{otherwise.}
  \end{cases}
\end{align}
If $\Luk$ and $\Gog$ are clear from the context, we omit the subscripts and write
just $\otimes$ and $\rightarrow$, respectively.
Notice that both $\otimes_\Luk$ and $\otimes_\Gog$ are continuous functions.
As a consequence,
\begin{align}
  \textstyle\bigwedge\{a \otimes b_i;\, i \in I\} &=
  a \otimes \textstyle\bigwedge\{b_i;\, i \in I\}
\end{align}
for all $a \in [0,1]$ and $\{b_i \in [0,1];\, i \in I\}$.
Moreover, $\rightarrow_\Luk$ is continuous in both its arguments,
$\rightarrow_\Gog$ is continuous in the second argument and
left-continuous in the first one. As a consequence,
the residua in both $\Luk$ and $\Gog$ satisfy
\begin{align}
  \textstyle\bigwedge\{a \rightarrow b_i;\, i \in I\} &=
  a \rightarrow \textstyle\bigwedge\{b_i;\, i \in I\},
  \\
  \textstyle\bigwedge\{a_i \rightarrow b;\, i \in I\} &=
  \textstyle\bigvee\{a_i;\, i \in I\} \rightarrow b,
  \\
  \textstyle\bigvee\{a \rightarrow b_i;\, i \in I\} &=
  a \rightarrow \textstyle\bigvee\{b_i;\, i \in I\}
  \label{eqn:VEEabi=aVEEbi}
\end{align}
and, in addition, $\Luk$ satisfies
\begin{align}
  \textstyle\bigvee\{a_i \rightarrow b;\, i \in I\} &=
  \textstyle\bigwedge\{a_i;\, i \in I\} \rightarrow b.
\end{align}
More details on the structures of degrees can be found in~\cite{KMP:TN}.

In order to simplify our considerations about if-then rules between
graded attributes, we use $\mathbf{L}$-fuzzy sets and related notions.
Consider a non-empty set $U$ which acts as a universe of elements.
Each map $A\!: U \to L$, where $L$ is the set of truth degrees in $\mathbf{L}$
is called an $\mathbf{L}$-fuzzy set~\cite{Gog:LFS}
(shortly, an $\mathbf{L}$-set) in $U$ and the degree $A(u)$ is interpreted as
``the degree to which $u$ belongs to the $\mathbf{L}$-set $A$''.
The collection of all $\mathbf{L}$-sets
in $U$ is denoted by $L^U$.
Furthermore, $A \in L^U$ is called finite whenever $\{u \in U;\, A(u) > 0\}$
is a finite set; $A \in L^U$ is called a singleton, written $\{{}^{a}/u\}$,
whenever $A(u) = a$ and $A(v) = 0$ for all $v \in U \setminus \{u\}$;
$A \in L^U$ is called rational whenever $\{A(u);\, u \in U\} \subseteq \RatL$
where $\RatL$ denotes the rational unit interval.

We consider operations with $\mathbf{L}$-sets which are defined componentwise
using the operations in $\mathbf{L}$. Namely, for $A_i \in L^U$ ($i \in I$),
$A \in L^U$, and $c \in L$, we define
$\bigcap \{A_i;\, i \in I\}$ (the intersection of $A_i$'s),
$\bigcup \{A_i;\, i \in I\}$ (the union of $A_i$'s),
$c \otimes A$ (the $c$-multiple of $A$), and
$c \rightarrow A$ (the $c$-shift of $A$) by putting
\begin{align}
  \textstyle\bigl(\bigcap \{A_i;\, i \in I\}\bigr)(u) &=
  \textstyle\bigwedge\{A_i(u);\, i \in I\},
  \label{eqn:op_intersection} \\
  \textstyle\bigl(\bigcup \{A_i;\, i \in I\}\bigr)(u) &=
  \textstyle\bigvee\{A_i(u);\, i \in I\},
  \label{eqn:op_union} \\
  (c {\otimes} A)(u) &= c \otimes A(u),
  \label{eqn:op_c_multiple} \\
  (c {\rightarrow} A)(u) &= c \rightarrow A(u),
  \label{eqn:op_c_shift}
\end{align}
for all $u \in U$. In addition, if $|I| = 2$, we use the usual
infix notation $A \cap B$ and $A \cup B$ to denote~\eqref{eqn:op_intersection}
and~\eqref{eqn:op_union}, respectively.

For $A,B \in L^U$, we consider two basic types of
containment relations. Namely, a bivalent containment relation and
a graded containment relation. In the first case, we write 
$A \subseteq B$ and say that $A$ is fully contained in $B$
whenever $A(u) \leq B(u)$ holds for each $u \in U$.
In the second case, we define a degree $S(A,B) \in L$
to which $A$ is a subset of $B$ by
\begin{align}
  S(A,B) = \textstyle\bigwedge\{A(u) \rightarrow B(u);\, u \in U\}.
  \label{eqn:S}
\end{align}
Using the properties of $\mathbf{L}$, it is easily seen that
$A \subseteq B$ if{}f $A(u) \rightarrow B(u) = 1$ holds for all $u \in U$
which is true if{}f $S(A,B) = 1$.
See~\cite[Theorem 3.12]{Bel:FRS}
for further details on properties of graded subsethood.

In the paper, we utilize particular fuzzy closure operators.
Recall that an operator $\C\!: L^U \to L^U$ is called an
$\mathbf{L}$-closure operator~\cite{Be:Fco,RoEsGaGo:Icoar}
on the set $U$ whenever the following conditions
\begin{align}
  A &\subseteq \C(A),
  \label{eqn:cl_ext} \\
  S(A,B) &\leq S(\C(A),\C(B)),
  \label{eqn:cl_mon} \\
  \C(\C(A)) &= \C(A),
  \label{eqn:cl_idm}
\end{align}
are satisfied for all $A,B \in L^U$,
cf. also~\cite[Section 7.1]{Bel:FRS}. A non-empty system
$\mathcal{S} \subseteq L^U$ (of $\mathbf{L}$-sets in $U$) is
called directed whenever for any $A,B \in \mathcal{S}$
there is $C \in \mathcal{S}$ such that $A \subseteq C$ and $B \subseteq C$.

\section{Rational fuzzy attribute logic}\label{sec:rfal}
In this section, we introduce the rational fuzzy attribute logic (RFAL).
We describe its language, formulas, their semantic entailment and a graded notion
of provability. We present the Pavelka completeness of our logic which is proved
in the next section.

The language of our logic is given by
\begin{itemize}
\item
  a (denumerable) set $\mathrm{Var}$ of \emph{propositional variables}, and
\item
  the set of \emph{constants for rational truth degrees}, i.e.,
  for each $a \in \RatL$ we consider a constant $\overline{a}$.
\end{itemize}
Propositional variables in $\mathrm{Var}$ are denoted by $p,q,\ldots$ Note that
the distinction between rational truth degrees in $[0,1]$ and the constants for
the degrees is essentially the same as in RPL~\cite{Haj:Flah} and is motivated
by proper distinction between the syntax and semantics of our logic.

Using symbols for logical connectives $\logand$ (conjunction) and
$\Rightarrow$ (implication), we consider formulas of the following form
\begin{align}
  ((\overline{a_1} \Rightarrow p_1) \logand \cdots \logand
  (\overline{a_m} \Rightarrow p_m))
  \Rightarrow
  ((\overline{b_1} \Rightarrow q_1) \logand \cdots \logand
  (\overline{b_n} \Rightarrow q_n)),
  \label{eqn:fml}
\end{align}
where $p_1,\ldots,p_m,q_1,\ldots,q_n \in \mathrm{Var}$ and 
$\overline{a_1},\ldots,\overline{a_m},\overline{b_1},\ldots,\overline{b_n}$ are
constants for truth degrees. We call such formulas
\emph{rational fuzzy attribute implications.} Note that according to the
standard meaning of $\logand$ and $\Rightarrow$ which is used in fuzzy logics
in the narrow sense, \eqref{eqn:fml} can be understood as a formula expressing
the fact that ``if $p_1$ is true at least to degree $a_1$ and ${\cdots}$ and
$p_n$ is true at least to degree $a_n$, then $q_1$ is true at least to degree
$b_1$ and ${\cdots}$ and $q_n$ is true at least to degree $b_n$'' which corresponds
with the intended meaning of the rules we have described in Section~\ref{sec:intro}.
Also note that~\eqref{eqn:fml} is a well-formed formula in a language of MTL (or
a stronger logic) which is extended by constants for truth degrees.

In order to define the semantics of formulas like~\eqref{eqn:fml}, we use the
usual notion of evaluation of propositional variables which uniquely extends to
all formulas, including those in the form of~\eqref{eqn:fml}. In a more detail,
we call any map $e\!: \mathrm{Var} \to [0,1]$ an \emph{evaluation} (notice that
$e(p)$ may be irrational) and define the degree $||\varphi||_e$ to which a general
well-formed formula in our language is true under $e$:
\begin{align}
  ||p||_e &= e(p), \\
  ||\varphi \Rightarrow \psi||_e &= ||\varphi||_e \rightarrow ||\psi||_e, \\
  ||\varphi \logand \psi||_e &= ||\varphi||_e \wedge ||\psi||_e, \\
  ||\overline{a}||_e &= a,
\end{align}
for all $p \in \mathrm{Var}$, formulas $\varphi,\psi$, and $a \in \RatL$
(other connectives may be introduces but we do not need them for our development,
cf.~\cite{EsGo:MTL,Haj:MFL}). If $\varphi$ is~\eqref{eqn:fml},
then $||\varphi||_e \in [0,1]$ is
the \emph{degree to which $\varphi$ is true under $e$.}
The notion of degrees of semantic entailment of formulas is then defined as
follows: An evaluation $e$ is called a \emph{model} of a set $\Sigma$ of formulas
whenever $||\varphi||_e = 1$ for all $\varphi \in \Sigma$.
The \emph{degree $||\varphi||_\Sigma$ to which $\varphi$ is entailed by $\Sigma$}
is defined
\begin{align}
  ||\varphi||_\Sigma &=
  \textstyle\bigwedge\{||\varphi||_e;\, e \text{ is a model of } \Sigma\}.
\end{align}
In this paper, we are primarily interested in syntactic characterization of 
$||\varphi||_\Sigma$ for $\varphi$ and all formulas in $\Sigma$ being of the
form~\eqref{eqn:fml}.

Before we introduce our axiomatization, let us present a concise way of representing
formulas like~\eqref{eqn:fml} and their entailment. Since both the antecedent and
consequent of~\eqref{eqn:fml} are of the form of conjunctions of subformulas
(and $\logand$ is interpreted by a truth function which is commutative,
associative, and idempotent), we may encode the antecedents and consequents
by finite rational $\mathbf{L}$-sets in $\mathrm{Var}$. Indeed,
for~\eqref{eqn:fml}, we may consider $A,B \in L^\mathrm{Var}$ given by
\begin{align}
  A(p) &=
  \begin{cases}
    a_i, &\text{if } p \text{ equals } p_i \text{ for } 1 \leq i \leq m, \\
    0, &\text{otherwise,}
  \end{cases}
  \label{eqn:Ap}
  \\
  B(p) &= 
  \begin{cases}
    b_i, &\text{if } p \text{ equals } q_i \text{ for } 1 \leq i \leq n, \\
    0, &\text{otherwise,}
  \end{cases}
  \label{eqn:Bp}
\end{align}
for all $p \in \mathrm{Var}$.
Under this notation, \eqref{eqn:fml} may be written as
\begin{align}
  A \Rightarrow B
  \label{eqn:abbrev}
\end{align}
and the degree $||A \Rightarrow B||_e$ to which $A \Rightarrow B$ is
true under $e$ may be defined as
\begin{align}
  ||A \Rightarrow B||_e &= S(A,e) \rightarrow S(B,e),
  \label{eqn:abbrev_sem}
\end{align}
where $S(A,e)$ and $S(B,e)$ are subsethood degrees~\eqref{eqn:S}.
A moment's reflection shows that if $\varphi$ denotes~\eqref{eqn:fml}
and $A \Rightarrow B$ is the corresponding abbreviation of~\eqref{eqn:fml}
with $A$ and $B$ given by~\eqref{eqn:Ap} and~\eqref{eqn:Bp}, respectively, then
\begin{align}
  ||\varphi||_e = ||A \Rightarrow B||_e
\end{align}
for any evaluation $e\!: \mathrm{Var} \to [0,1]$. Therefore, in the rest
of the paper, we use primarily the abbreviated form~\eqref{eqn:abbrev}
of formulas (where both $A,B$ are finite and rational $\mathbf{L}$-sets
in $\mathrm{Var}$) which simplifies some of our considerations.

\begin{remark}
  The first approach to attribute implications between graded (fuzzy) attributes
  was introduced by Pollandt~\cite{Po:FB}. She studied the formulas mainly from
  the point of view of formal concept analysis~\cite{GaWi:FCA} of data with fuzzy
  attributes and she did not present any complete axiomatization of the proposed
  semantic entailment. Conceptually, the formulas used by
  Pollandt~\cite{Po:FB} correspond to~\eqref{eqn:abbrev} and its interpretation
  as in~\eqref{eqn:abbrev_sem}. Later, the approach was generalized by considering
  truth-stressing hedges~\cite{BeVy:Fcalh,EsGoNo:Hedges,Haj:Ovt} as
  additional parameters of the interpretation of the formulas,
  see~\cite{BeVy:ICFCA} for a survey and \cite{BeVy:ADfDwG} for recent results.
  The utilization of hedges proved interesting
  because stronger (and desirable) properties of fuzzy attributes implications
  (like uniqueness of minimal bases obtained via pseudo-intents~\cite{BeVy:ADfDwG,GuDu})
  result by specific choices of hedges (e.g., by the choice of
  globalization~\cite{TaTi:Gist} which on linear structures of degrees coincides
  with Baaz's $\Delta$, cf.~\cite{Baaz}). Even more
  general concept of a parameterization of fuzzy attribute implications is
  introduced in~\cite{Vy:Pasefai} based on algebras of isotone Galois connections.
\end{remark}

From now on, we use the following assumption.

\begin{assumption}\label{ass:rat}
  $\mathbf{L} = \langle L,\wedge,\vee,\otimes,\rightarrow,0,1\rangle$ is
  a complete residuated lattice on the real unit interval and $\otimes$ and
  $\rightarrow$ are defined so that for any $a,b \in \RatL$, we have
  $a \otimes b \in \RatL$ and $a \rightarrow b \in \RatL$.
  In this case, we say that $\mathbf{L}$ is \emph{rationally closed.}
\end{assumption}

The assumption says basically that $\otimes$ and $\rightarrow$ applied to rational
arguments yield rational results. This is obviously satisfied for both $\Luk$ and
$\Gog$. As a consequence, if $A \in L^U$ and $c$ are rational, then $c {\otimes} A$
and $c {\rightarrow} A$ given by \eqref{eqn:op_c_multiple}
and \eqref{eqn:op_c_shift}, respectively, are also rational.

Our logic uses a deductive system which consists of a single axiom scheme
and two deduction rules:
\begin{enumerate}\parskip=0pt
\item[\itm{1}]
  each formula of the form $A{\cup}B \Rightarrow B$ where
  $A,B$ are finite rational $\mathbf{L}$-sets in $\mathrm{Var}$ is an \emph{axiom};
\item[\itm{2}]
  the deduction rules of \emph{cut} and \emph{multiplication}:
  \bgroup%
  \addtolength{\leftmarginii}{1em}%
  \begin{enumerate}\parskip=0pt
  \item[(Cut)]
    from $A \Rightarrow B$ and $B{\cup}C \Rightarrow D$
    infer $A{\cup}C \Rightarrow D$,
  \item[(Mul)]
    from $A \Rightarrow B$
    infer $c{\otimes}A \Rightarrow c{\otimes}B$,
  \end{enumerate}
  for all finite rational $\mathbf{L}$-sets $A,B,C,D$ in $\mathrm{Var}$
  and $c \in \RatL$.
  \egroup%
\end{enumerate}

Observe that $c{\otimes}A$ and $c{\otimes}B$ are both rational and finite, i.e.,
(Mul) is a well-defined deduction rule. Analogously, the axioms and (Cut) are well
defined because a union of finitely many rational finite $\mathbf{L}$-sets is
a finite rational $\mathbf{L}$-set.

\begin{remark}
  (a)
  Our deduction system resembles the famous system by Armstrong~\cite{Arm:Dsdbr}
  which in database systems plays a central role in reasoning about
  functional dependencies and normalized database schemes~\cite{Mai:TRD}.
  Namely, if we replace $\mathbf{L}$-sets by ordinary sets,
  the axioms together with (Cut)
  form a system which is equivalent to that of Armstrong. Note that in database
  systems (Cut) is usually presented under the name
  pseudotransitivity~\cite{Hol,Mai:TRD}. Deductive systems for fuzzy attribute
  implications based on (Mul) and (Cut) are also present
  in~\cite{BeVy:ICFCA,BeVy:ADfDwG} and generalized
  in \cite{BeVy:Falcrl,KuVy:Flprrai} by considering infinitary deduction rules
  to ensure completeness over arbitrary infinite $\mathbf{L}$.

  (b)
  Since the axioms and deduction rules we deal with use abbreviations for
  formulas like~\eqref{eqn:fml}, it is worth noting that our axioms and
  deduction rules may also be understood as formulas and deduction rules in
  the ordinary (narrow) sense. We may prove that in MTL (or a stronger logic)
  enriched by constants for rational truth degrees and bookkeeping axioms,
  the axiom schemes are provable, and (Mul) and (Cut) are derived deduction
  rules. In this remark, let $\mathcal{C}$ denote a logic which results from
  MTL (or a stronger logic) by adding constants for rational truth degrees
  and bookkeeping axioms which ensure that
  $\overline{a \rightarrow b} \Leftrightarrow \overline{a} \Rightarrow \overline{b}$,
  $\overline{a \otimes b} \Leftrightarrow \overline{a} \logtnorm \overline{b}$,
  $\overline{a \wedge b} \Leftrightarrow \overline{a} \loginf \overline{b}$,
  and
  $\overline{a \vee b} \Leftrightarrow \overline{a} \logsup \overline{b}$
  are all provable by $\mathcal{C}$. We inspect the following cases:
  \begin{itemize}
  \item[--]
    All axioms are provable by $\mathcal{C}$.
    First, observe that $A{\cup}B \Rightarrow B$
    can equivalently be written as $C \Rightarrow B$ where $C(p) \geq B(p)$
    for all $p \in \mathrm{Var}$.
    Now, fix $p$ and let $C(p) = a$ and $B(p) = b$. 
    Using the bookkeeping axioms together with $b \rightarrow a = 1$,
    we get
    
    $\vdash_\mathcal{C} \overline{b} \Rightarrow \overline{a}$

    and as a consequence of the transitivity of implication, we get

    $\vdash_\mathcal{C}
    (\overline{a} \Rightarrow p) \Rightarrow (\overline{b} \Rightarrow p)$.

    Now, we may repeat the idea finitely many times for all $p$
    with $C(p) > 0$ and utilize 

    $\vdash_\mathcal{C} (\varphi_1 \Rightarrow \psi_1) \Rightarrow
    ((\varphi_2 \Rightarrow \psi_2) \Rightarrow
    ((\varphi_1 \logand \varphi_2) \Rightarrow (\psi_1 \logand \psi_2)))$

    in order to show that $\mathcal{C}$ proves each axiom of our deductive system.

  \item[--]
    In case of (Cut), observe that we have

    $\{\varphi \Rightarrow \psi, (\psi \logand \chi) \Rightarrow \vartheta\}
    \vdash_{\mathcal{C}}
    (\varphi \logand \chi) \Rightarrow \vartheta$,

    where $\varphi,\psi,\chi,\vartheta$ are arbitrary formulas.
    In addition, if a formula of the form~\eqref{eqn:fml} is abbreviated by
    $B{\cup}C \Rightarrow D$, then it is equivalent to
    $(\psi \logand \chi) \Rightarrow \vartheta$ where 
    $\psi \Rightarrow \vartheta$ is abbreviated by $B \Rightarrow D$ and
    $\chi \Rightarrow \vartheta$ is abbreviated by $C \Rightarrow D$.
    Indeed, observe that $(\overline{a} \Rightarrow \varphi) \logand
    (\overline{c} \Rightarrow \varphi)$ is provably equivalent to 
    $(\overline{a} \logsup \overline{c}) \Rightarrow \varphi$, i.e.,
    $\overline{a \vee b} \Rightarrow \varphi$ by the bookkeeping axioms.
    As a consequence, the derived formula
    $(\varphi \logand \chi) \Rightarrow \vartheta$ may be abbreviated
    by $A{\cup}C \Rightarrow D$ provided that $\varphi \Rightarrow \psi$ is
    abbreviated by $A \Rightarrow B$. Altogether, the results of (Cut) are
    derivable in $\mathcal{C}$.

  \item[--]
    Finally, for any $\varphi$, $\psi$, and $c \in \RatL$, we have
    
    $\{\varphi \Rightarrow \psi\} \vdash_\mathcal{C}
    (\overline{c} \Rightarrow \varphi) \Rightarrow (\overline{c} \Rightarrow \psi)$.

    Furthermore, if $\varphi \Rightarrow \psi$ is of the form~\eqref{eqn:fml},
    then we may use

    $\vdash_{\mathcal{C}}
    (\overline{c} \Rightarrow (\varphi_1 \logand \cdots \logand \varphi_n))
    \Leftrightarrow
    ((\overline{c} \Rightarrow \varphi_1) \logand \cdots \logand
    (\overline{c} \Rightarrow \varphi_n))$.

    Thus, using the bookkeeping axioms together with

    $\vdash_{\mathcal{C}}
    (\varphi \Rightarrow (\psi \Rightarrow \chi)) \Leftrightarrow
    ((\varphi \logtnorm \psi) \Rightarrow \chi)$,

    we conclude that (Mul) is a derivable deduction rule in $\mathcal{C}$.
  \end{itemize}

  (c)
  Note that there are several deductive systems equivalent to the system
  of our axioms, (Cut), and (Mul). For instance, (Cut) and (Mul) can be
  reduced to a single deduction rule~\cite{BeVy:ADfDwG}. Alternatively,
  one may introduce normalized proofs based on the rule of accumulation
  analogously as in~\cite{BeVy:MRAP} (cf. also~\cite{Ma:Mcrdm}). Other
  deductive systems may involve the rule of simplification instead of
  (Cut), see~\cite{BeCoEnMoVy:SFL} (cf. also~\cite{SL} and
  a similar deduction rule proposed by Darwen in \cite{DaDa:RDW89_91}).
\end{remark}

Considering our deduction system, we introduce the notions of proofs and
provability degrees.
A \emph{proof} of $A \Rightarrow B$ by $\Sigma$ is any sequence of formulas
$C_1 \Rightarrow D_1,\ldots,C_n \Rightarrow D_n$ such that
$C_n = A$, $D_n = B$, and for each $i=1,\ldots,n$, we have that
\begin{enumerate}\parskip=0pt
\item[\itm{1}]
  $C_i \Rightarrow D_i$ is an axiom, or
\item[\itm{2}]
  $C_i \Rightarrow D_i \in \Sigma$, or
\item[\itm{3}]
  $C_i \Rightarrow D_i$ results from some
  of the formulas in $\{C_j \Rightarrow D_j;\, j < i\}$
  by a single application of the deduction rule (Cut) or (Mul).
\end{enumerate}
If there is a proof of $A \Rightarrow B$ by $\Sigma$, we denote the
fact by $\Sigma \vdash A \Rightarrow B$ and call $A \Rightarrow B$
\emph{provable} by $\Sigma$.
The \emph{degree $|A \Rightarrow B|_\Sigma$
  to which $A \Rightarrow B$ is provable by $\Sigma$} is defined by
\begin{align}
  |A \Rightarrow B|_\Sigma &=
  \textstyle\bigvee\{c \in \RatL;\, \Sigma \vdash A \Rightarrow c{\otimes}B\}.
  \label{eqn:synent}
\end{align}

\begin{remark}
  The provability degrees~\eqref{eqn:synent} are defined in much the same way as
  in RPL, i.e., via an ordinary notion of provability and without considering proofs
  as sequences of weighted formulas as in the original
  Pavelka approach~\cite{Pav:Ofl1,Pav:Ofl2,Pav:Ofl3}. Indeed, in RPL~\cite{Haj:Flah},
  the degrees of provability are introduced as
  \begin{align}
    |\varphi|^\mathrm{RPL}_\Sigma &=
    \textstyle\bigvee\{c \in \RatL;\,
    \Sigma \vdash_\mathrm{RPL} \overline{c} \Rightarrow \varphi\}.
  \end{align}
  Also note that if $\mathbf{L}$ is $\Luk$, RPL provides a syntactic characterization
  of the semantic entailment of the formulas considered in our logic.
  Indeed, for $\varphi \Rightarrow \psi$ abbreviated by $A \Rightarrow B$,
  one can easily see that
  \begin{align}
    ||A \Rightarrow B||_\Sigma &=
    |\varphi \Rightarrow \psi|^\mathrm{RPL}_\Sigma.
  \end{align}
  However, RPL-based proofs
  of $\overline{c} \Rightarrow (\varphi \Rightarrow \psi)$ by $\Sigma$ may
  contain formulas which are not of the form~\eqref{eqn:fml}.
  In contrast, our deductive system, allows to infer only
  (abbreviated representations of) formulas like~\eqref{eqn:fml}.
  As a consequence, the deduction in our system is simpler than the
  deduction in RPL and enables automated deduction as we shall see
  in Section~\ref{sec:details}. Analogous remarks can be made for
  R$\Pi$L (but notice that R$\Pi$L has an infinitary
  deduction rule~\cite{EsGoHaNa:Rflin}).
\end{remark}

The proof of the following assertion is elaborated in Section~\ref{sec:details}.

\begin{theorem}[Pavelka completeness of RFAL]\label{th:completeness}
  Let\/ $\mathbf{L}$ satisfy Assumption~\ref{ass:rat}
  and condition~\eqref{eqn:VEEabi=aVEEbi}.
  Then, for any $\Sigma$ and $A \Rightarrow B$,
  \begin{align}
    |A \Rightarrow B|_\Sigma = ||A \Rightarrow B||_\Sigma.
    \label{eqn:compl}
  \end{align}
  In particular, \eqref{eqn:compl} holds
  for $\mathbf{L}$ being $\Luk$ or $\Gog$.
  \qed
\end{theorem}

Using the results on RPL~\cite{Haj:MFL}, we may get further insight
into the previous completeness theorem which is based on our axiomatization.
Namely, \cite[Lemma 3.3.17]{Haj:MFL} gives that if $\Sigma$ is finite,
then each $||A \Rightarrow B||_\Sigma$ and thus $|A \Rightarrow B|_\Sigma$
is rational (for $\mathbf{L}$ being $\Luk$).
Let us note that the proof of \cite[Lemma 3.3.17]{Haj:MFL}
is technically quite involved. In Section~\ref{sec:details},
we show that in case of our logic, which considers only formulas
in the special form~\eqref{eqn:fml}, the argument is considerably easier.
The most important consequence of this property is summarized in
the following assertion which is also proved in Section~\ref{sec:details}.

\begin{theorem}\label{th:decidable}
  Let $\mathbf{L}$ be either of $\Luk$ or $\Gog$.
  Then, for each finite theory $\Sigma$ and any $A \Rightarrow B$,
  the degree $|A \Rightarrow B|_\Sigma$ is rational. In addition,
  we have
  \begin{align}
    &\Sigma \vdash A \Rightarrow c{\otimes}B
    &&\text{ for } c = |A \Rightarrow B|_\Sigma.
  \end{align}
  Moreover, RFAL based on either of\/ $\Luk$ and $\Gog$ is decidable.
  \qed
\end{theorem}

In particular, for $c=1$ Theorem~\ref{th:decidable} yields
$\Sigma \vdash A \Rightarrow B$ if{}f
$|A \Rightarrow B|_\Sigma = 1$.

\section{Proofs and Notes}\label{sec:details}
In this section, we give proofs of the basic assertions of RFAL presented in
the previous section. In the entire section, we assume that Assumption~\ref{ass:rat}
holds and that $\mathbf{L}$ satisfies~\eqref{eqn:VEEabi=aVEEbi}.

\begin{lemma}
  For any $A \Rightarrow B$ and rational $c \in [0,1]$, we have
  \begin{align}
    c \rightarrow |A \Rightarrow B|_\Sigma = |A \Rightarrow c{\otimes}B|_\Sigma.
    \label{eqn:c-shift_syn}
  \end{align}
\end{lemma}
\begin{proof}
  We check both inequalities of~\eqref{eqn:c-shift_syn}.
  
  In order to prove
  $c \rightarrow |A \Rightarrow B|_\Sigma \leq |A \Rightarrow c{\otimes}B|_\Sigma$,
  observe that by the definition of degrees of provability and
  using~\eqref{eqn:VEEabi=aVEEbi}, it follows that
  \begin{align}
    c \rightarrow |A \Rightarrow B|_\Sigma
    &= 
    c \rightarrow \textstyle\bigvee\{d \in [0,1]_\mathbb{Q};\,
    \Sigma \vdash A \Rightarrow d{\otimes}B\}
    \\
    &= 
    \textstyle\bigvee\{c \rightarrow d;\,
    d \in [0,1]_\mathbb{Q} \text{ and }
    \Sigma \vdash A \Rightarrow d{\otimes}B\}
    \\
    &\leq 
    \textstyle\bigvee\{c \rightarrow d;\,
    d \in [0,1]_\mathbb{Q} \text{ and }
    \Sigma \vdash A \Rightarrow (c \rightarrow d){\otimes}c{\otimes}B\}
    \\
    &\leq
    \textstyle\bigvee\{e \in [0,1]_\mathbb{Q};\,
    \Sigma \vdash A \Rightarrow e{\otimes}c{\otimes}B\}
  \end{align}
  because $\Sigma \vdash A \Rightarrow d{\otimes}B$ and
  $(c \rightarrow d) \otimes c \leq d$ yield 
  $\Sigma \vdash A \Rightarrow (c \rightarrow d){\otimes}c{\otimes}B$ by
  projectivity (from $\Sigma \vdash E \Rightarrow F{\cup}G$ one may derive that
  $\Sigma \vdash E \Rightarrow F$, see~\cite[Lemma 4.2]{BeVy:ADfDwG}), i.e.,
  $\Sigma \vdash A \Rightarrow e{\otimes}c{\otimes}B$
  for $e = c \rightarrow d$. Hence,
  \begin{align}
    c \rightarrow |A \Rightarrow B|_\Sigma
    &\leq
    \textstyle\bigvee\{e \in [0,1]_\mathbb{Q};\,
    \Sigma \vdash A \Rightarrow e{\otimes}(c{\otimes}B)\} \\
    &= |A \Rightarrow c{\otimes}B|_\Sigma.
  \end{align}

  Conversely, in order to prove
  $c \rightarrow |A \Rightarrow B|_\Sigma \geq |A \Rightarrow c{\otimes}B|_\Sigma$,
  it suffices to show
  $c \otimes |A \Rightarrow c{\otimes}B|_\Sigma \leq |A \Rightarrow B|_\Sigma$
  which is indeed the case:
  \begin{align}
    c \otimes |A \Rightarrow c{\otimes}B|_\Sigma
    &=
    c \otimes \textstyle\bigvee\{d \in [0,1]_\mathbb{Q};\,
    \Sigma \vdash A \Rightarrow d{\otimes}c{\otimes}B\} \\
    &=
    \textstyle\bigvee\{c \otimes d;\,
    d \in [0,1]_\mathbb{Q} \text{ and }
    \Sigma \vdash A \Rightarrow (c{\otimes}d){\otimes}B\} \\
    &\leq
    \textstyle\bigvee\{e \in [0,1]_\mathbb{Q};\,
    \Sigma \vdash A \Rightarrow e{\otimes}B\} \\
    &= |A \Rightarrow B|_\Sigma,
  \end{align}
  which concludes the proof.
\end{proof}

\begin{lemma}
  For any finite index set $I$ and an $I$-indexed set
  $\{B_i;\, i \in I\}$ of finite rational $\mathbf{L}$-sets in $\mathrm{Var}$,
  we have
  \begin{align}
    \textstyle\bigwedge\{|A \Rightarrow B_i|_\Sigma;\, i \in I\} =
    |A \Rightarrow \textstyle\bigcup\{B_i;\, i \in I\}|_\Sigma
    \label{eqn:union_syn}
  \end{align}
  for any finite rational $\mathbf{L}$-set $A$ in $\mathrm{Var}$.
\end{lemma}
\begin{proof}
  Observe that the claim is trivial for $I = \emptyset$ because
  \begin{align}
    \textstyle\bigwedge\emptyset
    = 1
    = |A \Rightarrow \bot|_\Sigma
    = |A \Rightarrow \textstyle\bigcup\emptyset|_\Sigma,
  \end{align}
  where $\bot(p) = 0$ for all $p \in \mathrm{Var}$. Therefore, we may only
  inspect the situation for non-empty finite $I$. Observe that
  $\textstyle\bigcup\{B_i;\, i \in I\}$ is finite and rational, i.e.,
  the formula $A \Rightarrow \textstyle\bigcup\{B_i;\, i \in I\}$ which appears
  on the right-hand side of~\eqref{eqn:union_syn} is well defined.
  We prove~\eqref{eqn:union_syn} by checking both inequalities.

  By definition of provability degrees and using the complete distributivity
  of the real unit interval (with its natural ordering, see~\cite{DaPr,Raney}),
  we have
  \begin{align}
    \textstyle\bigwedge\{|A \Rightarrow B_i|_\Sigma;\, i \in I\}
    &=
    \textstyle\bigwedge\{\bigvee\{c \in \RatL;\,
    \Sigma \vdash A \Rightarrow c{\otimes}B_i\};\, i \in I\}
    \\
    &=
    \textstyle\bigvee\{\bigwedge\{f(i);\, i \in I\};\, f \in \mathcal{F}\},
    \label{eqn:union1}
  \end{align}
  where
  \begin{align}
    \mathcal{F} &= \textstyle\prod\{\{c \in \RatL;\,
    \Sigma \vdash A \Rightarrow c{\otimes}B_i\};\, i \in I\}
  \end{align}
  is the direct product of subsets
  $\{c \in \RatL;\, \Sigma \vdash A \Rightarrow c{\otimes}B_i\}$ of $\RatL$.
  That is, $\mathcal{F}$ is the set of all maps (choice functions) of the form
  $f\!: I \to \RatL$
  such that 
  \begin{align}
    \Sigma \vdash A \Rightarrow f(i){\otimes}B_i
    \label{eqn:Sigma_f(i)}
  \end{align}
  for all $i \in I$. Since $I$ is finite, \eqref{eqn:Sigma_f(i)} yields
  \begin{align}
    \Sigma \vdash A \Rightarrow \textstyle\bigcup\{f(i){\otimes}B_i;\, i \in I\}
  \end{align}
  by additivity (from $\Sigma \vdash E \Rightarrow F$ and
  $\Sigma \vdash E \Rightarrow G$ it follows that
  $\Sigma \vdash E \Rightarrow F{\cup}G$,
  see~\cite[Lemma 4.2]{BeVy:ADfDwG}). Furthermore, $\{f(i);\, i \in I\}$
  has the least element $\bigwedge\{f(i);\, i \in I\}$ because $\mathbf{L}$
  is linear and $I$ is non-empty and finite. Therefore, by projectivity,
  we get
  \begin{align}
    \Sigma \vdash A \Rightarrow
    \textstyle\bigcup\{\bigwedge\{f(i);\, i \in I\}{\otimes}B_i;\, i \in I\}.
  \end{align}
  Hence, using the distributivity of $\otimes$ over general $\bigcup$,
  it follows that
  \begin{align}
    \Sigma \vdash A \Rightarrow
    \textstyle c{\otimes}\bigcup\{B_i;\, i \in I\}
  \end{align}
  for $c = \bigwedge\{f(i);\, i \in I\} \in \RatL$.
  Therefore, \eqref{eqn:union1} may be extended as
  \begin{align}
    \textstyle\bigwedge\{|A \Rightarrow B_i|_\Sigma;\, i \in I\}
    &=
    \textstyle\bigvee\{\bigwedge\{f(i);\, i \in I\};\, f \in \mathcal{F}\}
    \\
    &\leq
    \textstyle\bigvee\{c \in \RatL;\,
    \Sigma \vdash A \Rightarrow
    \textstyle c{\otimes}\bigcup\{B_i;\, i \in I\}\}
    \\
    &=
    \textstyle|A \Rightarrow \textstyle\bigcup\{B_i;\, i \in I\}|_\Sigma,
  \end{align}
  which proves the ``$\leq$''-inequality of~\eqref{eqn:union_syn}.

  Conversely, using projectivity and distributivity of $\otimes$ over
  general $\bigcup$:
  \begin{align}
    \textstyle|A \Rightarrow \textstyle\bigcup\{B_i;\, i \in I\}|_\Sigma
    &=
    \textstyle\bigvee\{c \in \RatL;\,
    \Sigma \vdash A \Rightarrow c{\otimes}\bigcup\{B_i;\, i \in I\}\}
    \\
    &=
    \textstyle\bigvee\{c \in \RatL;\,
    \Sigma \vdash A \Rightarrow \bigcup\{c{\otimes}B_i;\, i \in I\}\}
    \\
    &\leq
    \textstyle\bigvee\{c \in \RatL;\,
    \Sigma \vdash A \Rightarrow c{\otimes}B_i\}
    \\
    &=
    |A \Rightarrow B_i|_\Sigma
  \end{align}
  for arbitrary $i \in I$ which proves the converse inequality of~\eqref{eqn:union_syn}.
\end{proof}

\begin{lemma}
  For any $A \Rightarrow B$, $B \Rightarrow C$, $A \Rightarrow C$, and $\Sigma$,
  \begin{align}
    |A \Rightarrow B|_\Sigma \otimes |B \Rightarrow C|_\Sigma \leq
    |A \Rightarrow C|_\Sigma.
    \label{eqn:o_trans}
  \end{align}
\end{lemma}
\begin{proof}
  Let $\Sigma \vdash A \Rightarrow b{\otimes}B$ and
  $\Sigma \vdash B \Rightarrow c{\otimes}C$ for some $b,c \in \RatL$.
  Using (Mul), we infer
  $\Sigma \vdash b{\otimes}B \Rightarrow b{\otimes}c{\otimes}C$
  from $\Sigma \vdash B \Rightarrow c{\otimes}C$. Therefore, using (Cut)
  on $\Sigma \vdash A \Rightarrow b{\otimes}B$ and 
  $\Sigma \vdash b{\otimes}B \Rightarrow b{\otimes}c{\otimes}C$, we get
  $\Sigma \vdash A \Rightarrow b{\otimes}c{\otimes}C$. As a consequence,
  \begin{align}
    |A \Rightarrow B|_\Sigma \otimes |B \Rightarrow C|_\Sigma &=
    \textstyle\bigvee\{b \otimes c;\,
    \Sigma \vdash A \Rightarrow b{\otimes}B
    \text{ and }
    \Sigma \vdash B \Rightarrow c{\otimes}C\}
    \\
    &\leq
    \textstyle\bigvee\{b \otimes c;\,
    \Sigma \vdash A \Rightarrow b{\otimes}c{\otimes}C
    \}
    \\
    &= 
    \textstyle\bigvee\{d \in \RatL;\,
    \Sigma \vdash A \Rightarrow d{\otimes}C
    \}
    \\
    &= 
    |A \Rightarrow C|_\Sigma,
  \end{align}
  which proves~\eqref{eqn:o_trans}. 
\end{proof}

\begin{lemma}\label{le:compl}
  If $\Sigma \nvdash A \Rightarrow c{\otimes}B$ for $c \in \RatL$,
  then $||A \Rightarrow B||_{\Sigma} \leq c$.
\end{lemma}
\begin{proof}
  We assume $\Sigma \nvdash A \Rightarrow c{\otimes}B$ and we find
  an evaluation $e$ which is a model of $\Sigma$ and satisfies 
  $||A \Rightarrow B||_e \leq c$. For every $p \in \mathrm{Var}$, put
  \begin{align}
    e(p) &= |A \Rightarrow \{{}^{1\!}/p\}|_\Sigma.
    \label{eqn:e(p)}
  \end{align}
  Next, we prove an auxiliary claim: For each finite rational $\mathbf{L}$-set
  $C$ in $\mathrm{Var}$, we have $S(C,e) = |A \Rightarrow C|_\Sigma$ where $e$
  is given by~\eqref{eqn:e(p)}. Using the fact that $C$ may be expressed by
  a union of finitely many rational singletons in $\mathrm{Var}$ and
  using~\eqref{eqn:c-shift_syn} together with \eqref{eqn:union_syn},
  it follows that
  \begin{align}
    S(C,e)
    &=
    \textstyle\bigwedge\{C(p) \rightarrow e(p);\, p \in \mathrm{Var}\} \\
    &=
    \textstyle\bigwedge\{C(p) \rightarrow
    |A \Rightarrow \{{}^{1\!}/p\}|_\Sigma;\, p \in \mathrm{Var}\} \\
    &=
    \textstyle\bigwedge\{|A \Rightarrow C(p){\otimes}\{{}^{1\!}/p\}|_\Sigma;\,
    p \in \mathrm{Var}\} \\
    &=
    \textstyle\bigwedge\{|A \Rightarrow \{{}^{C(p)\!}/p\}|_\Sigma;\,
    p \in \mathrm{Var}\} \\
    &=
    \textstyle\bigwedge\{|A \Rightarrow \{{}^{C(p)\!}/p\}|_\Sigma;\,
    p \in \mathrm{Var} \text{ and } C(p) > 0\} \\
    &=
    |A \Rightarrow \textstyle\bigcup\{\{{}^{C(p)\!}/p\};\,
    p \in \mathrm{Var} \text{ and } C(p) > 0\}|_\Sigma \\
    &=
    |A \Rightarrow C|_\Sigma.
  \end{align}
  Using the claim, we prove that $e$ is a model of $\Sigma$.
  Indeed, take any $E \Rightarrow F \in \Sigma$ and observe
  that~\eqref{eqn:o_trans} yields
  \begin{align}
    S(E,e)
    &=
    |A \Rightarrow E|_\Sigma \\
    &=
    |A \Rightarrow E|_\Sigma \otimes 1 \\
    &=
    |A \Rightarrow E|_\Sigma \otimes |E \Rightarrow F|_\Sigma \\
    &\leq 
    |A \Rightarrow F|_\Sigma \\
    &=
    S(F,e),
  \end{align}
  showing $||E \Rightarrow F||_e = 1$ owing to~\eqref{eqn:abbrev_sem}.
  Furthermore, we have
  \begin{align}
    S(A,e) &= |A \Rightarrow A|_\Sigma = 1,
    \label{eqn:SAe=1}
  \end{align}
  because $A \Rightarrow A$ is an axiom. Furthermore, if $d > c$
  for $d \in \RatL$, then our assumption $\Sigma \nvdash A \Rightarrow c{\otimes}B$
  yields $\Sigma \nvdash A \Rightarrow d{\otimes}B$ because otherwise
  the assumption $\Sigma \nvdash A \Rightarrow c{\otimes}B$ would be violated
  on account of projectivity. Therefore,
  \begin{align}
    S(B,e) = |A \Rightarrow B|_\Sigma \leq c
    \label{eqn:Seb_leq_c}
  \end{align}
  As a consequence of~\eqref{eqn:SAe=1} and~\eqref{eqn:Seb_leq_c}, we get
  \begin{align}
    ||A \Rightarrow B||_\Sigma
    &\leq
    ||A \Rightarrow B||_e \\
    &=
    S(A,e) \rightarrow S(B,e) \\
    &\leq 
    1 \rightarrow c \\
    &= c,
  \end{align}
  which is the desired inequality.
\end{proof}

Theorem~\ref{th:completeness} can be now proved as follows:

\begin{proof}[Proof of Theorem~\ref{th:completeness}]
  The inequality $|A \Rightarrow B|_\Sigma \leq ||A \Rightarrow B||_\Sigma$,
  i.e., Pavelka-style soundness, follows by standard arguments. In order to prove
  the completeness in Pavelka style, we prove that for each $c \in \RatL$
  such that $c < ||A \Rightarrow B||_\Sigma$,
  we have $\Sigma \vdash A \Rightarrow c{\otimes}B$ which immediately gives
  $||A \Rightarrow B||_\Sigma \leq |A \Rightarrow B|_\Sigma$. But this is
  a direct consequence of Lemma~\ref{le:compl}.
\end{proof}

Since the degrees of semantic entailment of formulas of the form~\eqref{eqn:fml}
in RPL and R$\Pi$L are defined as in our logic, we immediately get the following
consequence on the relationship of to these two Pavelka-style complete logics:

\begin{corollary}
  The following are consequences of Theorem~\ref{th:completeness}:
  \begin{itemize}\parskip=0pt%
    \item
      If\/ $\mathbf{L}$ is $\Luk$, then RFAL is a fragment of RPL.
    \item
      If\/ $\mathbf{L}$ is $\Gog$, then RFAL is a fragment of $\mathit{R\Pi L}$.
      \qed
    \end{itemize}
\end{corollary}

\begin{remark}
  (a)
  Recall that R$\Pi$L utilizes an infinitary deduction rule~\cite{EsGoHaNa:Rflin}
  to ensure completeness in Pavelka style. In contrast, our logic for 
  $\mathbf{L}$ being $\Gog$, has the usual finitary notion of a proof.
  In other words, we have shown that R$\Pi$L restricted just to formulas
  of the form~\eqref{eqn:fml} can be axiomatized without infinitary deduction
  rules.

  (b)
  RFAL is not Pavelka-style complete with $\mathbf{L}$ being the standard G\"odel
  algebra (i.e., $\mathbf{L}$ defined on the real unit interval with
  $\otimes = \wedge$). Observe that for
  \begin{align}
    \Sigma &=
    \bigl\{\{{}^{0\!}/p\} {\,\Rightarrow\,} \{{}^{a\!}/p\};\,
    a \in [0,0.5)_\mathbb{Q}\bigr\} \cup
    \bigl\{\{{}^{0.5\!}/p\} {\,\Rightarrow\,} \{{}^{1\!}/q\}\bigr\},
  \end{align}
  we obviously have $||\{{}^{0\!}/p\} {\,\Rightarrow\,} \{{}^{1\!}/q\}||_\Sigma = 1$
  for $\mathbf{L}$ being the standard G\"odel algebra because for each model
  $e$ of $\Sigma$, we have $e(p) \geq 0.5$ and thus $e(q) = 1$. On the other hand,
  we have $|\{{}^{0\!}/p\} {\,\Rightarrow\,} \{{}^{1\!}/q\}|_\Sigma < 1$.
  Indeed, each finite $\Sigma' \subseteq \Sigma$ admits a rational model $e$
  such that $e(p) = a < 0.5$ and $e(q) = a$. Owing to soundness,
  we get $|\{{}^{0\!}/p\} {\,\Rightarrow\,} \{{}^{1\!}/q\}|_{\Sigma'} \leq a < 0.5$
  and thus $|\{{}^{0\!}/p\} {\,\Rightarrow\,} \{{}^{1\!}/q\}|_{\Sigma} \leq 0.5 < 1$.
  Therefore, in order to ensure completeness in Pavelka style for $\mathbf{L}$
  being the standard G\"odel algebra, one has to resort to infinitary
  deduction rules as in~\cite{BeVy:Falcrl,KuVy:Flprrai}.

  (c)
  Note for readers familiar with Pavelka's abstract logic as it is presented
  in~\cite[Section 9.2]{Haj:MFL}: Our approach uses the traditional understanding
  of theories as sets of formulas. Following Pavelka's ideas, one may extend it
  to use theories considered as $\mathbf{L}$-sets of formulas. This approach is,
  however, reducible to the traditional one: For $\Sigma$ considered as
  an $\mathbf{L}$-set which for any $A \Rightarrow B$ prescribes a general
  degree $\Sigma(A \Rightarrow B) \in [0,1]$ (including irrational degrees),
  we may consider
  \begin{align*}
    \Sigma^* &= \{A \Rightarrow c{\otimes}B;\,
    c \in \RatL \text{ and } c < \Sigma(A \Rightarrow B)\}
  \end{align*}
  which does the same job, see~\cite[Section 9.2]{Haj:MFL} for further details
  on semantic and syntactic entailment from ``fuzzy theories.''
\end{remark}

We now turn our attention to the proof of Theorem~\ref{th:decidable}.
We utilize a characterization of degrees of provability based on constructing
least models containing given evaluations.
For each evaluation $e$ and a theory $\Sigma$, we consider evaluations
$e^n_\Sigma$ ($n$ is a finite ordinal) and $e^\omega_\Sigma$ ($\omega$ denotes
the least infinite ordinal) as follows:
\begin{align}
  e^0_\Sigma &= e,
  \label{eqn:e_0}
  \\
  e^{n+1}_\Sigma &= e^n_\Sigma \cup
  \textstyle\bigcup\{S(A,e^n_\Sigma) {\otimes} B;\; A \Rightarrow B \in \Sigma\},
  \label{eqn:e_n}
  \\
  e^\omega_\Sigma &= \textstyle\bigcup\{e^n_\Sigma;\, n < \omega\},
  \label{eqn:e_omega}
\end{align}
for all $n < \omega$. Observe that if $\Sigma$ is finite and the evaluation
$e$ is finite and rational (recall that $e$ is in fact an $\mathbf{L}$-set),
then $e^n_\Sigma$ is finite and rational for each $n < \omega$ provided that
$\mathbf{L}$ satisfies Assumption~\ref{ass:rat}. The following assertions
show properties of $e^\omega_\Sigma$ given by~\eqref{eqn:e_omega}.

\begin{lemma}\label{le:least_model}
  Let\/ $\mathbf{L}$ satisfy~\eqref{eqn:VEEabi=aVEEbi}, $\Sigma$ be a theory
  and $e$ be an evaluation. Then, $e^\omega_\Sigma$ given
  by~\eqref{eqn:e_omega} is the least model of $\Sigma$ which contains $e$.
\end{lemma}
\begin{proof}
  We first show that $e^\omega_\Sigma$ is a model of $\Sigma$. Directly
  by~\eqref{eqn:e_n}, for each formula $A \Rightarrow B \in \Sigma$ and
  each $n < \omega$, we have
  \begin{align}
    S(A,e^n_\Sigma) {\otimes} B \subseteq e^{n+1}_\Sigma.
  \end{align}
  Therefore,
  \begin{align}
    \textstyle\bigcup\{S(A,e^n_\Sigma) {\otimes} B;\, n < \omega\}
    \subseteq
    \textstyle\bigcup\{e^{n+1}_\Sigma;\, n < \omega\} =
    e^\omega_\Sigma.
    \label{eqn:least_incl}
  \end{align}
  In addition, using \eqref{eqn:S}, \eqref{eqn:e_omega},
  and \eqref{eqn:VEEabi=aVEEbi}, it follows that
  \begin{align}
    S(A,e^\omega_\Sigma) \otimes B &=
    \textstyle\bigwedge\{A(p) \rightarrow e^\omega_\Sigma(p);\,
    p \in \mathrm{Var}\} \otimes B
    \\
    &=
    \textstyle\bigwedge\{A(p) \rightarrow \bigvee\{e^n_\Sigma(p);\,
    n < \omega\};\, p \in \mathrm{Var}\} \otimes B
    \\
    &=
    \textstyle\bigwedge\{\bigvee\{A(p) \rightarrow e^n_\Sigma(p);\,
    n < \omega\};\, p \in \mathrm{Var}\} \otimes B.
    \label{eqn:wedge_vee}
  \end{align}
  Let $\mathcal{F}$ denote the set of all maps
  of the form $f\!: \mathrm{Var} \to \{n;\, n < \omega\}$.
  Using complete distributivity,
  the previous equality can be extended as 
  \begin{align}
    S(A,e^\omega_\Sigma) \otimes B
    &=
    \textstyle\bigvee\{\bigwedge\{A(p) \rightarrow e^{f(p)}_\Sigma(p);\,
    p \in \mathrm{Var}\};\, f \in \mathcal{F}\} \otimes B.
    \label{eqn:vee_wedge_f}
  \end{align}
  In addition, $A$ is a finite $\mathbf{L}$-set. Therefore,
  there are only finitely many $p \in \mathrm{Var}$ such that
  $A(p) \rightarrow e^{f(p)}_\Sigma(p) < 1$. As a consequence,
  for any $f \in \mathcal{F}$, there is $n < \omega$ such that
  \begin{align}
    \textstyle\bigwedge\{A(p) \rightarrow e^{f(p)}_\Sigma(p);\,
    p \in \mathrm{Var}\} \leq
    \textstyle\bigwedge\{A(p) \rightarrow e^n_\Sigma(p);\,
    p \in \mathrm{Var}\}.
  \end{align}
  Hence, \eqref{eqn:vee_wedge_f} is equivalent to
  \begin{align}
    S(A,e^\omega_\Sigma) \otimes B
    &=
    \textstyle\bigvee\{\bigwedge\{A(p) \rightarrow e^n_\Sigma(p);\,
    p \in \mathrm{Var}\};\, n < \omega\} \otimes B
    \\
    &=
    \textstyle\bigvee\{S(A,e^n_\Sigma);\, n < \omega\} \otimes B
    \\
    &=
    \textstyle\bigcup\{S(A,e^n_\Sigma) {\otimes} B;\, n < \omega\}
    \\
    &\subseteq
    e^\omega_\Sigma
  \end{align}
  using~\eqref{eqn:least_incl}. Hence, we have shown
  $S(A,e^\omega_\Sigma) {\otimes} B \subseteq e^\omega_\Sigma$
  for any $A \Rightarrow B \in \Sigma$, i.e.,
  \begin{align}
    S(A,e^\omega_\Sigma) \leq S(B,e^\omega_\Sigma)
  \end{align}
  for any $A \Rightarrow B \in \Sigma$ and so
  $e^\omega_\Sigma$ is a model of $\Sigma$.

  Now, consider a model $e'$ of $\Sigma$ such that $e \subseteq e'$, i.e.,
  $e(p) \leq e'(p)$ for all $p \in \mathrm{Var}$. Then, by induction,
  $e^{n}_\Sigma \subseteq e'$ gives
  \begin{align}
    S(A,e^{n}_\Sigma) \otimes B \subseteq S(A,e') \otimes B \subseteq e'
  \end{align}
  for all $A \Rightarrow B \in \Sigma$, i.e., \eqref{eqn:e_n} immediately
  yields $e^{n+1}_\Sigma \subseteq e'$. As a consequence,
  $e^\omega_\Sigma \subseteq e'$, showing that
  $e^\omega_\Sigma$ is indeed the least model of $\Sigma$ containing $e$.
\end{proof}

A description of provability degrees based on least models is now established.
Indeed, using the existing result described in \cite[Theorem~3.11]{BeVy:ADfDwG},
it follows that $||A \Rightarrow B||_\Sigma = S(B,e)$ for $e$ being
the least model of $\Sigma$ containing~$A$.
Applying Theorem~\ref{th:completeness} and Lemma~\ref{le:least_model}, we
get the following characterization.

\begin{corollary}\label{th:least_compl}
  Let\/ $\mathbf{L}$ satisfy Assumption~\ref{ass:rat}
  and condition~\eqref{eqn:VEEabi=aVEEbi}.
  Then, for any $\Sigma$ and $A \Rightarrow B$,
  \begin{align}
    |A \Rightarrow B|_\Sigma &= S(B, A^\omega_\Sigma).
  \end{align}
\end{corollary}
\begin{proof}
  Consequence of Theorem~\ref{th:completeness}, Lemma~\ref{le:least_model},
  and \cite[Theorem~3.11]{BeVy:ADfDwG}.
\end{proof}

In case of $\Luk$ and $\Gog$, the observation in Lemma~\ref{le:least_model}
and its consequence in Corollary~\ref{th:least_compl} may further be improved
provided that $\Sigma$ is finite in which case one may always take some
$e^n_\Sigma$ instead of $e^\omega_\Sigma$:

\begin{lemma}\label{le:fin}
  Let $\mathbf{L}$ be either of $\Luk$ or $\Gog$. Then,
  for any finite theory $\Sigma$ and a rational evaluation $e$ there
  is $n$ such that $e^\omega_\Sigma = e^n_\Sigma$.
\end{lemma}
\begin{proof}
  We inspect the situation for $\Luk$ and $\Gog$ separately. In the proof, we
  denote by $O$ the set of all propositional variables which appear in formulas
  in $\Sigma$. That is, $p \in O$ whenever there is 
  $A \Rightarrow B \in \Sigma$ such that $A(p) > 0$ or $B(p) > 0$.
  Analogously, by $K$ we denote the set of non-zero degrees (represented by
  constants) which appear in formulas in $\Sigma$. That is, $c \in K$
  whenever there is $A \Rightarrow B \in \Sigma$ such that
  $A(p) = c > 0$ or $B(p) = c > 0$ for some $p \in \mathrm{Var}$.

  Note that since $\Sigma$ is finite, then $O$ is finite as well. As a consequence,
  there are at most finitely many $p \in \mathrm{Var}$ such that
  $e^\omega_\Sigma(p) > e(p)$. Namely, $e^\omega_\Sigma(p) > e(p)$ implies
  that $p \in O$ which follows directly by~\eqref{eqn:e_n}. Therefore, in the
  proof, we tacitly assume that $e(p) = 0$ whenever $p \not\in O$.

  Let $\mathbf{L}$ be $\Luk$ and observe that since $K$ is finite and all degrees
  in $K$ are rational, then one may consider a finite equidistant \L ukasiewicz
  subchain of $\Luk$ consisting of rational truth degrees which contains all
  the degrees from $K$ and $\{e(p);\, p \in \mathrm{Var}\}$. Denote the subchain
  by $\mathbf{L}'$ and its support set by $L'$.
  Then, inspecting \eqref{eqn:e_0} and~\eqref{eqn:e_n}, it follows
  that each $e^n_\Sigma$ is rational and
  \begin{align}
    \{e^n_\Sigma(p);\, p \in \mathrm{Var}\} \subseteq L'
  \end{align}
  because each $e^n_\Sigma(p)$ may be expressed by finitely many applications
  of operations in $\mathbf{L}'$.
  As a consequence, $e^\omega_\Sigma = e^n_\Sigma$
  for some $n < \omega$ because only finitely may propositional variables
  (those in $O$) are assigned finitely many different degrees (those in $L'$), i.e.,
  the chain $e^0_\Sigma \subseteq e^1_\Sigma \subseteq e^2_\Sigma \subseteq \cdots$
  consists of only finitely many proper inclusions.

  For $\mathbf{L}$ being $\Gog$, we proceed by contradiction: We assume
  there are $p \in \mathrm{Var}$ such that $e^n_\Sigma(p) < e^\omega_\Sigma(p)$
  for all $n < \omega$. Since $\Sigma$ and $O$ are finite, there are
  pairwise distinct $p_0,\ldots,p_{m-1} \in \mathrm{Var}$ and formulas
  $A_0 \Rightarrow B_0,\ldots,A_{m-1} \Rightarrow B_{m-1}$ in $\Sigma$
  such that for $p_m = p_0$ and for infinitely many $n < \omega$, we have
  \begin{align}
    e^{n+1}_\Sigma(p_1) &=
    (A_0(p_0) \rightarrow e^{n}_\Sigma(p_0)) \otimes B_0(p_1),
    \label{eqn:inf1}
    \\
    e^{n+2}_\Sigma(p_2) &=
    (A_1(p_1) \rightarrow e^{n+1}_\Sigma(p_1)) \otimes B_1(p_2),
    \\
    \vdots\phantom{(p_2)} &\kern3em\vdots \notag
    \\
    e^{n+m}_\Sigma(p_m) &=
    (A_{m-1}(p_{m-1}) \rightarrow e^{n+m-1}_\Sigma(p_{m-1})) \otimes B_{m-1}(p_m),
    \label{eqn:inf2}
  \end{align}
  and the following conditions are satisfied:
  \begin{itemize}
  \item[--]
    $e^{n+i}_\Sigma(p_i) > e^{n+i-1}_\Sigma(p_i)$ for all $i=1,\ldots,m$, and
  \item[--]
    $A_i(p_i) > e^{n+i}_\Sigma(p_i)$ for all $i=0,\ldots,m-1$.
  \end{itemize}
  As a consequence, $e^{n+m}_\Sigma(p_0) = e^{n+m}_\Sigma(p_m) > e^{n}_\Sigma(p_0)$.
  Furthermore, directly by the definition of $\rightarrow$ in $\Gog$, we get that
  \begin{align}
    e^{n+i}_\Sigma(p_i) &=
    \frac{e^{n+i-1}_\Sigma(p_{i-1}) \cdot B_{i-1}(p_i)}{A_{i-1}(p_{i-1})}
  \end{align}
  for all $i=1,\ldots,m$. Therefore, $e^{n+m}_\Sigma(p_0)$ may be expressed as
  \begin{align}
    e^{n+m}_\Sigma(p_0) = e^{n+m}_\Sigma(p_m) &=
    \frac{e^n_\Sigma(p_0) \cdot B_0(p_1) \cdot B_1(p_2) \cdots B_{m-1}(p_m)}{
      A_0(p_0) \cdot A_1(p_1) \cdots A_{m-1}(p_{m-1})}.
  \end{align}
  That is, we have $e^{n+m}_\Sigma(p_0) = e^n_\Sigma(p_0) \cdot c$ for $c > 1$.
  Since we have assumed that \eqref{eqn:inf1}--\eqref{eqn:inf2} hold for infinitely
  many $n$, then there are $n_1,n_2,\ldots$ such that
  \begin{align}
    e^{n_{i+1}}_\Sigma(p_0) \geq e^{n_i+m}_\Sigma(p_0) = e^{n_i}_\Sigma(p_0) \cdot c
  \end{align}
  for all $i=1,2,\ldots$ and thus
  \begin{align}
    e^{n_{i}}_\Sigma(p_0) \geq e^{n_1}_\Sigma(p_0) \cdot c^{i-1}
  \end{align}
  for all $i=1,2,\ldots$ which means that
  \begin{align}
    \lim_{i \to \infty}e^{n_{i}}_\Sigma(p_0) \geq
    \lim_{i \to \infty}e^{n_1}_\Sigma(p_0) \cdot c^{i-1} =
    e^{n_1}_\Sigma(p_0) \lim_{i \to \infty}c^{i-1} =
    \infty,
  \end{align}
  which contradicts the fact that $e^n_\Sigma(p_0) < e^{\omega}_\Sigma(p_0)$
  for all $n < \omega$.
\end{proof}

\begin{lemma}\label{le:prov_An}
  Let $\Sigma$ be finite. Then, $\Sigma \vdash A \Rightarrow A^n_\Sigma$
  for all $n < \omega$.
\end{lemma}
\begin{proof}
  The proof goes by induction. Notice that since $\Sigma$ is finite,
  then $A^n_\Sigma$ is rational and finite for any $n < \omega$, i.e.,
  $A \Rightarrow A^n_\Sigma$ is well defined formula.
  Suppose that $\Sigma \vdash A \Rightarrow A^n_\Sigma$
  and take $E \Rightarrow F \in \Sigma$. Then, using (Mul) for $c = S(E,A^n_\Sigma)$,
  we have
  $\Sigma \vdash S(E,A^n_\Sigma){\otimes}E \Rightarrow
  S(E,A^n_\Sigma){\otimes}F$.
  Then, using $S(E,A^n_\Sigma){\otimes}E \subseteq A^n_\Sigma$ and (Cut),
  we get $\Sigma \vdash A \Rightarrow S(E,A^n_\Sigma){\otimes}F$.
  Inspecting~\eqref{eqn:e_n}, we prove $\Sigma \vdash A \Rightarrow A^{n+1}_\Sigma$
  by finitely many applications of additivity
  (from $\Sigma \vdash E \Rightarrow F$ and
  $\Sigma \vdash E \Rightarrow G$ we derive
  $\Sigma \vdash E \Rightarrow F{\cup}G$,
  see~\cite[Lemma 4.2]{BeVy:ADfDwG}) because $\Sigma$ is finite.
\end{proof}

\begin{proof}[Proof of Theorem~\ref{th:decidable}]
  Let $\mathbf{L}$ be either of $\Luk$ or $\Gog$ and
  let $\Sigma$ be finite. Taking into account Corollary~\ref{th:least_compl}
  and Lemma~\ref{le:fin}, we get
  \begin{align}
    |A \Rightarrow B|_\Sigma &= 
    S(B, A^\omega_\Sigma) = 
    S(B, A^n_\Sigma)
  \end{align}
  for some $n < \omega$. Since both $B$ and $A^n_\Sigma$ are finite
  and rational, $S(B, A^n_\Sigma)$ is a rational degree which proves that
  $|A \Rightarrow B|_\Sigma$ is rational. In addition, Lemma~\ref{le:prov_An}
  yields that $\Sigma \vdash A \Rightarrow A^n_\Sigma$. Hence, using
  the fact that $S(B,A^n_\Sigma){\otimes}B \subseteq A^n_\Sigma$,
  the projectivity gives $\Sigma \vdash A \Rightarrow c{\otimes}B$
  for $c = S(B, A^n_\Sigma) = |A \Rightarrow B|_\Sigma$.

  In addition, it is decidable
  whether $|A \Rightarrow B|_\Sigma = c$: In finitely many steps,
  one computes $A^n_\Sigma$ such that $A^n_\Sigma = A^{n+1}_\Sigma$
  (which exists owing to Lemma~\ref{le:fin}) and checks
  whether $c = S(B, A^n_\Sigma)$. In particular, for $c = 1$,
  it means $B \subseteq A^n_\Sigma$ if{}f $|A \Rightarrow B|_\Sigma = 1$
  if{}f $\Sigma \vdash A \Rightarrow 1{\otimes}B$ if{}f 
  $\Sigma \vdash A \Rightarrow B$, i.e., it is decidable whether
  $A \Rightarrow B$ is provable by $\Sigma$.
\end{proof}

\section{Rational $\mathbf{L}$-closure Operators}\label{sec:clops}
In this section, we present observations on closure structures associated
to models of theories consisting of formulas of the form~\eqref{eqn:fml}.
We are motivated by the fact that for any $\Sigma$, we may introduce an
operator which maps each evaluation $e$ to $e^\omega_\Sigma$ defined
by~\eqref{eqn:e_omega}. Under the assumption of~\eqref{eqn:VEEabi=aVEEbi},
we show that such operators are in fact finitary $\mathbf{L}$-closure
operators on propositional variables.

Recall from preliminaries the general notion of an $\mathbf{L}$-closure operator,
see Section~\ref{sec:prelim}.
We call an $\mathbf{L}$-closure operator $\C\!: L^U \to L^U$
\emph{finitary} whenever
\begin{align}
  \C(A) &= \textstyle\bigcup\{\C(B);\,
  B \subseteq A \text{ and } B \text{ is finite}\}
  \label{eqn:finitary}
\end{align}
for all $A \in L^U$. For $A,B \in L^U$, we put
\begin{align}
  B \sqsubseteq A
\end{align}
whenever $B$ is rational, finite, and $B \subseteq A$.
Using this notation, we call an $\mathbf{L}$-closure operator $\C\!: L^U \to L^U$
\emph{rational} whenever
\begin{align}
  \C(A) &= \textstyle\bigcup\{\C(B);\, B \sqsubseteq A\}
  \label{eqn:rational_finitary}
\end{align}
for all $A \in L^U$. By definition, a rational $\mathbf{L}$-closure operator
is finitary. The following assertions show that under
the condition~\eqref{eqn:VEEabi=aVEEbi},
finitary $\mathbf{L}$-closure operators are always rational.

\begin{lemma}\label{le:cup}
  Let $A$ be a finite $\mathbf{L}$-set in $U$ and let $\mathcal{B}$ be
  a directed system of $\mathbf{L}$-sets in $U$. If\/ $\mathbf{L}$ defined
  on the real unit interval satisfies~\eqref{eqn:VEEabi=aVEEbi}, then
  \begin{align}
    \textstyle S(A,\bigcup\mathcal{B}) &=
    \textstyle\bigvee\{S(A,B);\, B \in \mathcal{B}\}.
    \label{eqn:cup}
  \end{align}
\end{lemma}
\begin{proof}
  Note that since $\mathcal{B}$ is directed, it is also non-empty.
  Let $\mathcal{F}$ denote the set of all functions from $\mathrm{Var}$
  to $\mathcal{B}$. Under this notation, using~\eqref{eqn:VEEabi=aVEEbi}
  and the complete distributivity, we get
  \begin{align}
    S(A,\textstyle\bigcup\mathcal{B})
    &=
    \textstyle\bigwedge
    \{A(p) \rightarrow \textstyle(\bigcup\mathcal{B})(p);\, 
    p \in \mathrm{Var}\}
    \\
    &=
    \textstyle\bigwedge
    \{A(p) \rightarrow \textstyle\bigvee\{B(p);\, B \in \mathcal{B}\};\, 
    p \in \mathrm{Var}\}
    \\
    &=
    \textstyle\bigwedge
    \{\textstyle\bigvee\{A(p) \rightarrow B(p);\, B \in \mathcal{B}\};\, 
    p \in \mathrm{Var}\}
    \\
    &=
    \textstyle\bigvee
    \{\textstyle\bigwedge\{A(p) \rightarrow (f(p))(p);\, p \in \mathrm{Var}\};\, 
    f \in \mathcal{F}\}.
  \end{align}
  Now, since $A$ is finite and $\mathcal{B}$ is directed,
  for each $f \in \mathcal{F}$
  there is $B \in \mathcal{B}$ such that $f(p) \subseteq B$
  for all $p \in \mathrm{Var}$ satisfying $A(p) > 0$ and thus
  $(f(p))(p) \leq B(p)$ for all $p \in \mathrm{Var}$
  satisfying $A(p) > 0$. Therefore, we have
  \begin{align}
    S(A,\textstyle\bigcup\mathcal{B})
    &=
    \textstyle\bigvee
    \{\textstyle\bigwedge\{A(p) \rightarrow (f(p))(p);\, p \in \mathrm{Var}\};\, 
    f \in \mathcal{F}\}
    \\
    &=
    \textstyle\bigvee
    \{\textstyle\bigwedge\{A(p) \rightarrow B(p);\, p \in \mathrm{Var}\};\, 
    B \in \mathcal{B}\}
    \\
    &=
    \textstyle\bigvee
    \{S(A,B);\, B \in \mathcal{B}\},
  \end{align}
  which establishes~\eqref{eqn:cup}.
\end{proof}

\begin{lemma}\label{le:epsilon}
  Let $A$ be a finite $\mathbf{L}$-set in $U$. If\/ $\mathbf{L}$ defined on the
  real unit interval satisfies~\eqref{eqn:VEEabi=aVEEbi},
  then for any $\varepsilon < 1$ there is $B_\varepsilon \sqsubseteq A$
  such that $S(A,B_\varepsilon) > \varepsilon$.
\end{lemma}
\begin{proof}
  Consider $\mathbf{L}$-sets $B_n$ ($n < \omega$)
  with $B_n \subseteq B_{n+1}$ and $B_n \sqsubseteq A$
  for all $n < \omega$ such that 
  and $\bigcup\{B_n;\, n < \omega\} = A$.
  Since $\{B_n;\, n < \omega\}$ is obviously directed,
  Lemma~\ref{le:cup} yields
  \begin{align}
    1
    &=
    S(A,\textstyle\bigcup\{B_n;\, n < \omega\})
    =
    \textstyle\bigvee\{S(A,B_n);\, n < \omega\}.
  \end{align}
  Hence, for every $\varepsilon < 1$ there is $B_\varepsilon \sqsubseteq A$
  such that $S(A,B_\varepsilon) > \varepsilon$ otherwise our observation
  $1 = \textstyle\bigvee\{S(A,B_n);\, n < \omega\}$ would be violated.
\end{proof}

\begin{lemma}\label{le:fin_rational}
  Let $\mathbf{L}$ satisfy~\eqref{eqn:VEEabi=aVEEbi}. Then,
  every finitary $\mathbf{L}$-closure operator is rational.
\end{lemma}
\begin{proof}
  Let $\C\!: L^U \to L^U$ be an $\mathbf{L}$-closure operator which is finitary.
  It suffices to show that for any finite $\mathbf{L}$-set $A$ in $U$,
  we have $\C(A) = \bigcup\{\C(B);\, B \sqsubseteq A\}$.
  Based on the observation in Lemma~\ref{le:epsilon},
  for each $n < \omega$ we let $B_n \in L^U$
  such that $B_n \sqsubseteq A$ and $S(A,B_n) > 1 - \frac{1}{n}$.
  Applying the monotony~\eqref{eqn:cl_mon} of $\C$, for any $n < \omega$,
  we have $S(A,B_n) \otimes \C(A) \subseteq \C(B_n)$, i.e.,
  \begin{align}
    \C(A) &= 
    1 \otimes \C(A) =
    \textstyle\bigvee\{S(A,B_n) ;\, n < \omega\} \otimes \C(A)
    \\
    &=
    \textstyle\bigcup\{S(A,B_n) \otimes \C(A);\, n < \omega\}
    \\
    &\subseteq
    \textstyle\bigcup\{\C(B_n);\, n < \omega\}.
  \end{align}
  The converse inclusion holds trivially. As a consequence, $\C$ is rational.
\end{proof}

\begin{remark}
  Let us note that the assumption of $\mathbf{L}$ satisfying
  \eqref{eqn:VEEabi=aVEEbi} in Lemma~\ref{le:fin_rational} is essential.
  Indeed, in general there are finitary $\mathbf{L}$-closure operators which
  are not rational. For instance, let $\mathbf{L}$ be the standard G\"odel algebra,
  consider $U = \{u\}$, and let $c$ be an irrational number in $[0,1]$. Put
  \begin{align}
    \C(\{{}^{a\!}/u\})(u) &=
    \begin{cases}
      a, &\text{for } a < c, \\
      1, &\text{otherwise.}
    \end{cases}
  \end{align}
  Obviously, $\C$ satisfies \eqref{eqn:cl_ext} and \eqref{eqn:cl_idm}. In order
  to see that $\C$ satisfies \eqref{eqn:cl_mon}, observe that in the non-trivial
  case for $\{{}^{a\!}/u\}$ and $\{{}^{b\!}/u\}$ with $a > b$, we have
  $S(\{{}^{a\!}/u\},\{{}^{b\!}/u\}) = b$. Now, if $a \geq c$ and $b < c$, we get
  \begin{align}
    S(\C(\{{}^{a\!}/u\}),\C(\{{}^{b\!}/u\})) = 1 \rightarrow b = b =
    S(\{{}^{a\!}/u\},\{{}^{b\!}/u\})
  \end{align}
  If both $a \geq c$ and $b \geq c$, the condition is trivial;
  the same applies if $a < c$ and $b < c$. Altogether, $\C$ is
  an $\mathbf{L}$-closure operator. Furthermore,
  \begin{align}
    \textstyle\bigvee\{\C(\{{}^{a\!}/u\})(u);\, a \in \RatL \text{ and } a \leq c\} =
    c < 1 = \C(\{{}^{c\!}/u\})(u),
  \end{align}
  i.e., $\C$ is finitary but it is not rational.
\end{remark}

Now, for any $\Sigma$ and evaluation $e$, we put
\begin{align}
  \C[\Sigma](e) &= e^\omega_\Sigma.
  \label{eqn:c}
\end{align}
The following assertions characterize operators defined as in~\eqref{eqn:c}
for all possible choices of $\Sigma$.

\begin{theorem}\label{th:fp1}
  Let $\mathbf{L}$ satisfy~\eqref{eqn:VEEabi=aVEEbi}. Then, for each $\Sigma$,
  $\C[\Sigma]\!: L^\mathrm{Var} \to L^\mathrm{Var}$ defined by~\eqref{eqn:c}
  is a rational $\mathbf{L}$-closure operator.
\end{theorem}
\begin{proof}
  The fact that $\C[\Sigma]$ is an $\mathbf{L}$-closure operator follows
  by the general result in~\cite[Theorem 3.9]{BeVy:ADfDwG} which holds for
  any complete residuated lattice $\mathbf{L}$ taken as the structure of degrees.
  Also, the ``$\supseteq$''-part of \eqref{eqn:rational_finitary} is trivial.
  Thus, it suffices
  to check the ``$\subseteq$''-part of \eqref{eqn:rational_finitary}.
  Let $\mathcal{G} = \{g \in L^\mathrm{Var};\, g \sqsubseteq e\}$.
  We proceed by checking that
  $\textstyle\bigcup \{\C[\Sigma](g);\, g \in \mathcal{G}\}$ is
  a model of $\Sigma$ containing $e$. Observe that $\mathcal{G}$
  is directed and so is $\{\C[\Sigma](g);\, g \in \mathcal{G}\}$ owing
  to the monotony~\eqref{eqn:cl_mon} of $\C$.
  Take any $A \Rightarrow B \in \Sigma$.
  Applying Lemma~\ref{le:cup} and the fact that for each $g \in \mathcal{G}$,
  $\C[\Sigma](g)$ is a model of $\Sigma$ and so
  $S(A,\C[\Sigma](g)) \otimes B \subseteq \C[\Sigma](g)$,
  it follows that
  \begin{align}
    S(A,\textstyle\bigcup\{\C[\Sigma](g);\, g \in \mathcal{G}\}) \otimes B
    &= 
    \textstyle\bigvee \{S(A,\C[\Sigma](g));\, g \in \mathcal{G}\} \otimes B 
    \\
    &=
    \textstyle\bigcup \{S(A,\C[\Sigma](g)) \otimes B;\, g \in \mathcal{G}\} 
    \\
    &\subseteq
    \textstyle\bigcup \{\C[\Sigma](g);\, g \in \mathcal{G}\},
  \end{align}
  i.e., $\textstyle\bigcup \{\C[\Sigma](g);\, g \in \mathcal{G}\}$
  is indeed a model of $\Sigma$ which obviously contains $e$.
\end{proof}

\begin{theorem}\label{th:fp2}
  Let $\mathbf{L}$ satisfy~\eqref{eqn:VEEabi=aVEEbi} and let
  $\C: L^\mathrm{Var} \to L^\mathrm{Var}$ be
  a finitary $\mathbf{L}$-closure operator.
  Then, there is $\Sigma$ such that $\C = \C_\Sigma$.
\end{theorem}
\begin{proof}
  Let $\top \in L^\mathrm{Var}$ such that $\top(p) = 1$
  for all $p \in \mathrm{Var}$. We put
  \begin{align}
    \Sigma &= \{A \Rightarrow B;\,
    A \sqsubseteq \top \text{ and } B \sqsubseteq \C(A)\}
  \end{align}
  and prove the claim holds for $\Sigma$.

  First, we show that for any evaluation $e$,
  $\C(e)$ is a model of $\Sigma$, i.e., $\C(e)$ is closed under $\C[\Sigma]$.
  Let $A \Rightarrow B \in \Sigma$. Since $B \subseteq \C(A)$,
  \eqref{eqn:cl_mon} and~\eqref{eqn:cl_idm} yield
  \begin{align}
    S(A,\C(e)) &\leq S(\C(A),\C(\C(e)))
    \\
    &= S(\C(A),\C(e))
    \\
    &\leq S(B,\C(e)),
  \end{align}
  showing that $\C(e)$ is a model of $\Sigma$.

  Conversely, we show that $\C[\Sigma](e)$ is a fixed point of $\C$.
  Take any $A \sqsubseteq \C[\Sigma](e)$.
  Trivially, $S(A,\C[\Sigma](e)) = 1$ and thus for any
  $B \sqsubseteq \C(A)$, we have $S(B,\C[\Sigma](e)) = 1$, i.e.,
  $B \subseteq \C[\Sigma](e)$ because $A \Rightarrow B \in \Sigma$ and
  $\C[\Sigma](e)$ is a model of $\Sigma$. Now, the fact that 
  $B \subseteq \C[\Sigma](e)$ holds for all $B \sqsubseteq \C(A)$ yields that
  \begin{align}
    \C(A) =
    \textstyle\bigcup\{B;\, B \sqsubseteq \C(A)\} \subseteq \C[\Sigma](e).
  \end{align}
  Furthermore, the previous inclusion holds for any $A \sqsubseteq \C[\Sigma](e)$,
  i.e., utilizing the fact that $\C$ is rational which follows by
  Lemma~\ref{le:fin_rational}, we have
  \begin{align}
    \C(\C[\Sigma](e)) =
    \textstyle\bigcup\{\C(A);\, A \sqsubseteq \C[\Sigma](e)\}
    \subseteq \C[\Sigma](e),
  \end{align}
  proving that $\C[\Sigma](e)$ is a fixed point of $\C$.
\end{proof}

Theorem~\ref{th:fp1} and Theorem~\ref{th:fp2} showed that under the
assumption~\eqref{eqn:VEEabi=aVEEbi}, the operators
on $\mathbf{L}$-sets of propositional variables defined by~\eqref{eqn:c}
are exactly all rational $\mathbf{L}$-closure operators on
propositional variables. In other words, the systems of fixed points of
rational $\mathbf{L}$-closure operators on propositional variables
are the systems of models of sets of rational fuzzy attribute implications,
cf. \cite{Vy:Faiep} for a study of the expressive power of general fuzzy
attribute implications parameterized by hedges.


\subsubsection*{Conclusion}
Logic for reasoning with graded if-then rules is proposed. The rules can be seen
as formulas of the form of implications which contain constants for rational truth
degrees. The interpretation of formulas is given by complete residuated lattices
defined on the real unit interval. Degrees of semantic entailment and degrees of
provability are defined. For complete residuated lattices with residuum which
is continuous in the second argument, the logic is Pavelka-style complete which
means that degrees of semantic entailment agree with degrees of provability.
Characterization of the degrees of provability based on computing least models
is established. In case of finite theories and the standard \L ukasiewicz or
Goguen (product) algebras, the least models may be determined in finitely many
steps which shows that the logic based on these structures of degrees is
decidable. Structures of models are identified with fixed points of
rational $\mathbf{L}$-closure operators. It is shown that the property
of being rational is a consequence of the property of being finitary
in case of structures of truth degrees with residua continuous in
the second argument.

\subsubsection*{Acknowledgment}
Supported by grant no. \verb|P202/14-11585S| of the Czech Science Foundation.


\footnotesize
\bibliographystyle{amsplain}
\bibliography{rfal}

\end{document}